\documentclass[fleqn]{article}
\usepackage{amssymb,latexsym}
\usepackage{amsmath, amsthm}
\DeclareFontFamily{OT1}{rsfs}{} \DeclareFontShape{OT1}{rsfs}{m}{n}{
<-7> rsfs5 <7-10> rsfs7 <10-> rsfs10}{}
\DeclareMathAlphabet{\mycal}{OT1}{rsfs}{m}{n}
\newtheorem{Theorem}{Theorem}
\newcommand{\rd}{{\rm d}} 
\newcommand{\dtwo}{\Delta}
\newcommand{\iv}{\raisebox{-2.15pt}{$\lrcorner$}\hspace*{-3.3pt}\rfloor}
\newcommand{\lie}{{\pounds}} 
\begin{document}
\title{\bf On the existence of Kundt's metrics with
compact sections of null hypersurfaces}
\author{Jacek Jezierski\thanks{E-mail: Jacek.Jezierski@fuw.edu.pl}\\
Department of Mathematical Methods in Physics, \\ University of Warsaw,
ul. Ho\.za 74, 00-682 Warsaw, Poland}
\maketitle

\begin{abstract}
It is shown that Kundt's  metric for vacuum
cannot be constructed when two-dimensional space-like
sections of null hypersurfaces are compact, connected
manifolds with no boundary unless they are tori or spheres, i.e.
higher genus
$\mathbf{g} \geq 2$ is excluded by vacuum Einstein equations.
The so-called {\em basic equation} (resulting from Einstein equations)
is examined.
This is a non-linear PDE for unknown covector field and unknown
Riemannian structure on the two-dimensional manifold.
It implies several important results derived in this paper.
It arises not only
for Kundt's class but also for degenerate Killing horizons and
vacuum degenerate isolated horizons.

\end{abstract}

\section{Introduction}

Let us consider a null hypersurface in a Lorentzian spacetime $M$
which is a
three-dimensional submanifold $S \subset M$ such that the
restriction $g_{ab}$ of the spacetime metric $g_{\mu\nu}$ to $S$
is degenerate.

We shall often use adapted coordinates, where coordinate $x^3=u$ is
constant on $S$.  Coordinates on $S$ will be labeled by $a,b=0,1,2$
and sometimes coordinate $x^0$ will be denoted by $v$,
finally, coordinates on $B_{v} := \{x\in S \, | \; x^0=v=\mbox{const}\}$
 will be labeled by $A,B=1,2$.
Spacetime coordinates will be labeled by Greek characters
$\alpha, \beta, \mu, \nu$.

The non-degeneracy of the spacetime metric implies that the
metric $g_{ab}$ induced on $S$ from the spacetime metric
$g_{\mu\nu}$ has signature $(0,+,+)$. This means that there is a
non-vanishing null-like vector field $X^a$ on $S$, such that its
four-dimensional embedding $X^\mu$ to $M$ (in adapted coordinates
$X^3=0$) is orthogonal to $S$. Hence, the covector $X_\nu = X^\mu
g_{\mu\nu} = X^a g_{a\nu}$ vanishes on vectors tangent to $S$ and,
therefore, the following identity holds:
\begin{equation}\label{degeneracy}
  X^a g_{ab} \equiv 0 \ .
\end{equation}
It is easy to prove (cf. \cite{JKC}) that integral curves of
$X^a$, after a suitable reparameterization, are geodesic curves of
the spacetime metric $g_{\mu\nu}$. Moreover, any null
hypersurface $S$ may always be embedded in a one-parameter
congruence of null hypersurfaces.

We assume that topologically we have $S =  I \times
S^2$ where $I\subset{\mathbb R}^1$ is a real interval.
 Since our considerations are purely local, we fix the
orientation of the ${\mathbb R}^1$ component and assume that
null-like vectors $X$ describing degeneracy of the metric $g_{ab}$ of
$S$ will be always compatible with this orientation. Moreover, we
shall always use coordinates such that the coordinate $x^0$
increases in the direction of $X$, i.e.,~inequality $X(x^0) = X^0
> 0$ holds. In these coordinates degeneracy fields are of the form
$X = f(\partial_0-n^A\partial_A)$, where $f > 0$, $n_A = g_{0A}$
and we rise indices with the help of the two-dimensional matrix
${\tilde{\tilde g}}^{AB}$, inverse to $g_{AB}$.

If by $\lambda$ we denote the two-dimensional volume form on each
surface $x^0 = \mbox{\rm const}$:
\begin{equation}\label{lambda}
  \lambda:=\sqrt{\det g_{AB}} \ ,
\end{equation}
then for any degeneracy field $X$ of $g_{ab}$ the following
object
\[
v_{X} := \frac {\lambda}{X(x^0)}
\]
is a well defined scalar density on $S$
according to \cite{jjPRD02}.  This means
that ${\bf v}_X := v_X dx^0 \wedge dx^1 \wedge dx^2$ is a
coordinate-independent differential three-form on $S$. However,
$v_X$ depends upon the choice of the field  $X$.

It follows immediately from the above definition that the
following object:
\[
\Lambda = v_X \ X
\]
is a well defined (i.e.,~coordinate-independent) vector density on
$S$. Obviously, it {\em does not depend} upon any choice of the
field $X$:
\begin{equation}\label{Lambda}
\Lambda =  \lambda (\partial_0-n^A\partial_A) \ .
\end{equation}
Hence, it is an intrinsic property of the internal geometry
$g_{ab}$ of $S$. The same is true for the divergence $\partial_a
\Lambda^a$, which is, therefore, an invariant, $X$-independent,
scalar density on $S$. Mathematically (in terms of differential
forms), the quantity $\Lambda$ represents the two-form:
\[
{\bf L} := \Lambda^a \left( \partial_a \; \iv  \;  dx^0 \wedge
dx^1 \wedge dx^2 \right) \ ,
\]
whereas the divergence represents its exterior derivative (a
three-from): d${\bf L} := \left( \partial_a \Lambda^a \right)dx^0
\wedge dx^1 \wedge dx^2$. In particular, a null surface with
vanishing d${\bf L}$ is called a {\em non-expanding horizon} (see
\cite{ABL}).

The examples of spacetimes  obeying
Einstein equations suggest that non-expanding horizons are rather
isolated objects. In this paper we consider the problem of
  existence of one-parameter
congruence of local\footnote{We assume that coordinates
$u$ and $v$ are only in small neighborhood and
$M$ is constructed locally around given sphere.} non-expanding horizons.
The family of null hypersurfaces which are simultaneously
non-expanding horizons leads to the algebraically special
spacetimes so called non-diverging solutions or Kundt's class
of metrics (see chapter 31 in \cite{exact}).

\section{Topological rigidity}
Following chapter 31 of \cite{exact} let us consider
the line element in the following form:
\begin{equation}\label{km}
 \rd s^2 = g_{AB}\rd x^A \rd x^B - 2 \rd u \left(\rd v
 + m_A \rd x^A + H \rd u \right), \quad A,B =1,2
\end{equation}
or equivalently (see (31.6) in \cite{exact}) in a complex notation:
\begin{equation}\label{kmc}
\rd s^2 = 2P^{-2}\rd \zeta \rd\bar\zeta - 2 \rd u \left(\rd v
 + W \rd\zeta+\overline{W}\rd\bar\zeta+ H \rd u \right)
\end{equation}

We assume that two-dimensional section
(parameterized by coordinates $x^1,x^2$ or $\zeta,\bar\zeta$)
 of null hypersurfaces $u=\,$const
is a compact, connected
manifold $B$ with no boundary. The extrinsic curvature
$l_{ab}=-\frac12 \lie_X g_{ab}$ of the
null hypersurface $u=\,$const vanishes because $\partial_v g_{AB}=0$
(cf. \cite{JKC}).

\begin{Theorem}\label{main}
For any Riemannian metric $g_{AB}$
on a two-dimensional, compact, connected manifold
with no boundary and genus $\mathbf{g}\geq 2$ the vacuum Einstein
equations imply
no solutions for the line element (\ref{km})
which describes Kundt's class of metrics.
\end{Theorem}

\begin{proof}
Einstein-Maxwell equations for Kundt's metrics
split into system of non-linear two-dimensional partial differential
equations (eqs (31.21) in \cite{exact})
  \begin{eqnarray} \nonumber
&&\left(P^2 W_{,v}\right)_{,\zeta} -\frac12 \left(W_{,v}\right)^2=0\, ,  \\
\label{EM1} && \Phi_{1,\zeta}= W_{,v} \Phi_1 \, ,\\ \nonumber
&& \dtwo \ln P +\frac12 P^2 \left( \overline{W}_{,v\zeta}+{W}_{,v\bar\zeta}
-2W_{,v}\overline{W}_{,v}\right) = 2 \kappa_0 \left|\Phi_1\right|^2
  \end{eqnarray}
The remaining Einstein-Maxwell equations
(see 31.22-25 in \cite{exact}) reduce to polynomial dependence on $v$ and
linear problems on $B$ if we assume that $u$-dependence is given.

In the case of vacuum\footnote{It is enough to assume that
$\Phi_1=0$.} the equations (\ref{EM1}) imply
the following {\em basic equation} on $B$:
\begin{equation}
\label{basic}
  \omega_{A||B} +\omega_{B||A} +2\omega_A \omega_B = R_{AB} \, ,
\end{equation}
where $\omega_A$ corresponds to $\partial_v W$, $||$ denotes
covariant derivative with respect to the metric $g_{AB}$ and
$R_{AB}$ is its Ricci tensor. The equation (\ref{basic})
is a starting point of our considerations and it is a special case of
(3.7) in \cite{ABL}, if we assume that ${\tilde S}_{AB}$ vanishes.

The traceless part of (\ref{basic}) reads
\begin{equation}
\label{TSom}
  \omega_{A||B} +\omega_{B||A} - g_{AB}\omega^C{_{||C}}
  = -2\omega_A \omega_B + g_{AB} \omega^C \omega_C
\end{equation}
and for the trace we get
\begin{equation}
\label{divom}
  \omega^A{_{||A}} = K - \omega^A \omega_A \, ,
\end{equation}
where $K:=\frac12{\tilde{\tilde g}}^{AB}R_{AB}$ is
the Gaussian curvature of $B$
and ${\tilde{\tilde g}}^{AB}$ is the two-dimensional inverse metric.

Let us notice that eq. (\ref{divom}) and Gauss-Bonnet
theorem
\begin{equation}\label{GB}
 2-2\mathbf{g}=\frac1{2\pi}\int_B\lambda K
 \end{equation}
exclude immediately the case $\mathbf{g}\geq 2$ because
\[  0 \leq \int_B \lambda\omega_A\omega^A=\int_B\lambda K < 0 \]
which is impossible.
\end{proof}
Moreover, for $\mathbf{g}=1$ equation (\ref{divom}) and (\ref{GB}) imply
that on a torus the vector field $\omega^A$ is vanishing
and we obtain the following result:
\begin{Theorem}
\label{torus}
For any Riemannian metric $g_{AB}$
on a two-dimensional torus  
 equation (\ref{basic}) possesses only trivial
solutions $\omega^A\equiv 0 \equiv K$ and the metric $g_{AB}$
is flat.
\end{Theorem}
The Theorems \ref{main} and \ref{torus}
do not cover the most interesting case $B=S^2$.
We would like to add some more observations which are
valid in general case before we restrict ourselves to the case
when manifold $B$ is a sphere.

Contracting equation (\ref{TSom}) with $\omega^A\omega^B$,
we obtain the following identity:
\begin{equation}
\label{iom}
\omega^B\left(\omega^A \omega_A\right)_{||B}=
 \omega^A \omega_A \omega^B{_{||B}} -\left(\omega^A \omega_A\right)^2\, .
\end{equation}
Using (\ref{divom}) and (\ref{iom}), one can check the following
equality
\begin{equation}\label{3om}
  \left[ \omega^B\left(\omega^A \omega_A\right)^\alpha\right]_{||B} =
  -(2\alpha+1)\left(\omega^A \omega_A\right)^{\alpha+1}  +
  (\alpha+1)\left(\omega^A \omega_A\right)^\alpha K \, ,
\end{equation}
which finally implies one-parameter family of integral identities
\begin{equation}\label{intfg}
  \frac{2\alpha+1}{\alpha+1}\int_{B} \lambda F^{\alpha+1}=
  \int_{B} \lambda K F^\alpha \, ,
\end{equation}
where $F:=\omega^A \omega_A$ and $\lambda:=\sqrt{\det g_{AB}}$.

Suppose $\omega^A$ has only finite set of critical
points which are isolated.
 Then using eq. (\ref{3om}) for $\alpha=-\frac12$ we obtain
\begin{equation}\label{12K}
  \left[ {\omega^B\over
  \sqrt{\omega^A \omega_A}}\right]_{||B} =
  \frac12 \frac{K}{\sqrt{\omega^A \omega_A}} \, .
\end{equation}
Surrounding critical points of $\omega^A$ by small circles and
passing to the limit (i.e. shrinking circles to critical points)
 one can check that
\begin{equation}\label{Kom}
 \int_B \frac{\lambda K}{\sqrt{\omega^A \omega_A}} =0 \, ,
 \end{equation}
which is a special case of (\ref{intfg}) for $\alpha=-\frac12$.

Now, let us restrict ourselves to the case $B=S^2$.
From Gauss-Bonnet theorem we have
\[ \int_{S^2} \lambda K =4\pi > 0 \, .\]
Hence the condition (\ref{Kom}) implies that
$K$ must be negative on some open subset of $S^2$.
The above considerations can be summarized by the following
\begin{Theorem}\label{tw3}
There are no solutions of equation (\ref{basic})
with the following properties:
\begin{itemize}
\item $\omega^A=0$ only at finite set of points,
\item $B$ is a sphere with non-negative Gaussian curvature.
\end{itemize}
\end{Theorem}

When two-dimensional surface $B$ is a
 sphere then the corresponding null hypersurface
  is the non-expanding horizon.
A one-parameter family of non-expanding horizons
is still possible in the case of vacuum Einstein equations
but the Gaussian curvature has to be negative on some
domains or the vector field $\omega^A$ vanishes
on infinite set of points.

Let us observe that
\[  \int_{S^2} \lambda\omega_A\omega^A=\int_{S^2}\lambda K =4\pi  \]
implies that $\omega^A \neq 0$ on an open subset.
On the other hand, we show in the next Section an example of solutions
for Einstein-Maxwell equations
with $\omega^A \equiv 0$, which means that some
 arguments used above
are no longer true for non-vacuum solutions.

\section{Axially symmetric solutions for Einstein-Maxwell equations}
The two-dimensional part of Einstein-Maxwell equations (\ref{EM1})
can be written as follows
 \begin{eqnarray} \nonumber
&& \omega_{A||B} +\omega_{B||A} - g_{AB}\omega^C{_{||C}}
  = -2\omega_A \omega_B + g_{AB} \omega^C \omega_C \, ,  \\
&& f_{A||B} + f_{B||A} - g_{AB} f^C{_{||C}}
  = -2f_A \omega_B -2f_B \omega_A + 2 g_{AB} f^C \omega_C
\label{EM2}   \, ,\\ \nonumber
&& K - \omega^A{_{||A}}  - \omega^A \omega_A =  \kappa_0 f_A f^A \, ,
  \end{eqnarray}
where $f_A:=F_{vA}$ is a covector on $B$ corresponding
to $\Phi_1$. The all objects in (\ref{EM2}),
namely $\omega_A$, $f_A$ and $g_{AB}$, do not depend on $v$.

Some arguments from the proof of Theorem \ref{main}
 can be generalized to the
case with Maxwell field. In particular, the equation (\ref{3om})
takes now the following form:
\begin{equation}\label{3omEM}
  \left[ \omega^B\left(\omega^A \omega_A\right)^\alpha\right]_{||B} =
  -(2\alpha+1)\left(\omega^A \omega_A\right)^{\alpha+1}  +
  (\alpha+1)\left(\omega^A \omega_A\right)^\alpha
  \left(K-\kappa_0 f_A f^A\right) \, ,
\end{equation}
and we get the same integral identity (\ref{intfg}) for functions
$F=\omega^A \omega_A$ and $K-\kappa_0 f_A f^A$
instead of $K$. Moreover, for genus $\mathbf{g}\neq 0$ we can
repeat the arguments and we get
\begin{eqnarray} \nonumber
&& 4\pi(1-\mathbf{g}) = \int_B \lambda K =
\int_B \lambda\left(\omega^A \omega_A + \kappa_0 f_A f^A
\right) \, ,
  \end{eqnarray}
which implies $\mathbf{g}=1$, $\omega_A \equiv 0$,
$f_A \equiv 0$ and finally $K=0$.
This result can be described as follows
\begin{Theorem}\label{twEM}
For any Riemannian metric $g_{AB}$
on a two-dimensional, compact, connected manifold
with no boundary and genus $\;\mathbf{g}\geq 1$
equations (\ref{EM2})  have
no solutions for $\;\mathbf{g}\geq 2$ and for
$\mathbf{g}=1$ they possess only trivial solutions
i.e. $\omega_A \equiv 0$,
$f_A \equiv 0$ and flat metric $g_{AB}$.
\end{Theorem}

The Theorem \ref{twEM} together
with the assumption of non-triviality of $f_A$
restricts ourselves to $B=S^2$.
Moreover, let us assume that $\omega_A \equiv 0$. Hence,
equations (\ref{EM2}) reduce to the following
system of equations on $S^2$:
 \begin{eqnarray}
&& f_{A||B} + f_{B||A} - g_{AB} f^C{_{||C}}
  = 0
\label{cvf}   \, ,\\ \label{Kff}
&& K =  \kappa_0 f_A f^A
  \end{eqnarray}
The equation (\ref{cvf}) simply means that $f^A$ is a conformal
vector field on $S^2$. Moreover, the metric $g_{AB}$ is always conformally
related to a round unit sphere metric $h_{AB}$ i.e.
\[ g_{AB}=\Omega^2 h_{AB} \, , \quad \Omega > 0 \]
and eq. (\ref{Kff}) reduces to
\begin{equation}\label{cfo}
-\dtwo_h \ln\Omega = \kappa_0 \Omega^4 h_{AB}f^Af^B -1 \, ,
\end{equation}
where $\dtwo_h$ is the Laplace-Beltrami operator on $S^2$
with respect to the metric $h_{AB}$.

The construction of all axially symmetric solutions of
equations (\ref{cvf}) and (\ref{cfo}) can be obtained as follows:
In the coordinate system $(\theta,\phi)$ such that
$h_{AB}\rd x^A \rd x^B=\rd \theta^2 +\sin^2\theta \rd\varphi^2$
the axially symmetric solutions of (\ref{cvf}) belong to the following
two-dimensional family of conformal vector fields:
\[ f^\theta=-a\sin\theta\, , \quad f^\varphi=b \]
and $h_{AB}f^Af^B=(a^2+b^2)\sin^2\theta$. Hence,
if we assume $\partial_\varphi \Omega=0$,
the equation (\ref{cfo}) simplifies to the following form
\begin{equation} \label{cfu}
  {\rd \over\rd x}\left[ (1-x^2) {\rd \ln\Omega\over\rd x}\right]=
  1-d(1-x^2)\Omega^{4} \, ,
\end{equation}
where $d:=\kappa_0 (a^2+b^2)$ is a positive real constant
 and new coordinate $x:=\cos\theta$. A general solution of (\ref{cfu})
\[ \Omega^4 =
{4\beta^2 c(1-x^2)^{\beta-2}\over
d\left[2c(1+x)^\beta+(1-x)^\beta\right]^2} \]
becomes regular and positive for $\beta=2$ and
\[ \Omega^2={4 \over 2c(1+x)^2+(1-x)^2}\sqrt{\frac{c}d} \]
is an admissible conformal factor for any positive constant $c$.

The above result can be extended to the full space-time Einstein-Maxwell
solution in the Kundt's form\footnote{We remind that $\Omega$ corresponds
to $P$ and the vector field $f^A$ to $\Phi_1$.}
 similar to (31.57) in \cite{exact} or rather
to (31.55) with $G^0=\Phi^0_2=\partial_u(\ln P)=\partial_u \Phi_1=W=0$.
In this simple case the
equations (31.56) reduce to $\dtwo H^0 =0$ hence $H^0=\,$const.

\section{Other facts resulting from basic equation}
We start again with equation (\ref{basic}):
\begin{equation}
\label{basic1}
  \omega^A{_{||B}} +\omega_{B}{^{||A}} +2\omega^A \omega_B =
  R^A{_B}=\frac12 R\delta^A{_B} \, ,
\end{equation}
where $\omega^A$ is now a vector, $\omega_B=g_{AB}\omega^B$,
 $||$ denotes covariant derivative
with respect to the metric $g_{AB}$ and $R_{AB}$ is its Ricci
tensor. The above equation appears not only in the context of
Kundt's class, it also arises in the study of vacuum degenerate
isolated horizons \cite{ABL}, \cite{LP}, \cite{LPJJ}. Moreover, any
degenerate Killing horizon also implies this equation \cite{CRT}.
Hence, it is important to explore properties of this equation. We
already know that for axial symmetry and spherical topology there is
a unique solution -- extremal Kerr (see \cite{LPJJ}). Moreover, when
one-form $\omega_B\rd x^B$ is closed (e.g. static degenerate horizon
\cite{CRT}) there are no solutions of (\ref{basic}). However, in
general, the space of solutions is not known.

The traceless part of (\ref{basic1}) reads
\begin{equation}
\label{TSom1}
  \omega_{A||B} +\omega_{B||A} - g_{AB}\omega^C{_{||C}}
  = -2\omega_A \omega_B + g_{AB} \omega^C \omega_C
\end{equation}
and for the trace we get
\begin{equation}
\label{divom1}
  \omega^A{_{||A}} = K_g - \omega^A \omega_A \, ,
\end{equation}
where $K_g:=\frac12{\tilde{\tilde g}}^{AB}R_{AB}$ is the Gaussian
curvature of $g$ and ${\tilde{\tilde g}}^{AB}$ is the
two-dimensional inverse metric.

Let us notice that eq. (\ref{divom1}) enables one to perform
conformal transformation which leads to non-negative curvature.
More precisely,
let us choose $\alpha$ such that
\[ \triangle_g \alpha = \omega^A{_{||A}} \]
then from (\ref{divom1}) we get
\begin{equation}\label{Kl}
  K_g - \triangle_g \alpha =  \omega^A \omega^B g_{AB} \, ,
\end{equation}
Now, we define
\begin{equation}\label{hAB} h_{AB} := \exp(2\alpha)g_{AB}
\end{equation}
hence
\[ \exp(2\alpha)K_h = K_g - \triangle_g \alpha \]
and finally
\[ K_h = \exp(-4\alpha) \omega^A \omega^B h_{AB} \, \]
is non-negative.
Moreover, traceless part (\ref{TSom1}) is invariant
with respect to conformal rescaling (\ref{hAB}) of the metric $g$:
\[ \Gamma^A{_{BC}}(h)=\Gamma^A{_{BC}}(g)+\delta^A{_B}\alpha_{,C}
 + \delta^A{_C}\alpha_{,B} - g_{BC}\alpha_{,D}{\tilde{\tilde g}}^{DA} \]

\[ \omega^A{_{||B}} +\omega_{B}{^{||A}}-\delta^A{_B}\omega^{C}{_{||C}}=
   \nabla_B(h)\omega^A +\nabla^A(h)(h_{BC}\omega^{C})-\delta^A{_B}
   \nabla_C(h)\omega^{C} \]
hence we get
\begin{eqnarray} \label{TSh}
\nabla_B(h)\omega^A +\nabla^A(h)(h_{BC}\omega^{C})-\delta^A{_B}
   \nabla_C(h)\omega^{C} &=& -2\omega^A \omega_B
   + \delta^A{_B}\omega^C \omega_C \\ \nonumber
   &=& \exp(-2\alpha)
   \left(-2\omega^A h_{BC}\omega^C
   + \delta^A{_B}h_{CD}\omega^C \omega^D \right)
   \end{eqnarray}

Contracting equation (\ref{TSh}) with $\omega^A\omega^B$, we obtain
the following identity:
\begin{equation}
\label{hom} \omega^B\nabla_B\left(\|\omega\|^2\right)=
 \|\omega\|^2 \nabla_B\omega^B -
 \exp(-2\alpha)\left(\|\omega\|^2\right)^2\, ,
\end{equation}
where $\|\omega\|:=\sqrt{h_{AB}\omega^A\omega^B}$.

The last equality implies that when $\|\omega\|>0$ there are no
solutions of  equation (\ref{basic1}).
More precisely, we have:
\begin{equation}\label{hom1}
\nabla_B\left(\omega^B\|\omega\|^{-2}\right)=
 \exp(-2\alpha)=\frac{\sqrt{K_h}}{\|\omega\|} > 0 \, ,
\end{equation}
and integrating the above equality we get a contradiction.
This is not surprising because any vector field on a sphere
vanishes at least at one point.

\subsection{Integrability conditions}
Equation (\ref{basic1}) written in the following equivalent form:
\begin{equation}\label{basic2}
\omega_{A||B}= f\varepsilon_{AB}+\frac14 R g_{AB} -\omega_A\omega_B \, ,
\end{equation}
where $f:=\frac12\omega_{A||B}\varepsilon^{AB}$ is an unknown function on a sphere,
implies as follows:
\[ \omega_{A||BC}\varepsilon^{BC} =
   -f_{,A} - 3f\omega_A+\frac14\varepsilon_{AB}\left(
   R^{||B}+R\omega^B \right) \, . \]
Moreover, definition of curvature gives
\[ \omega_{A||BC}\varepsilon^{BC} = R^D{_{ABC}}\omega_D\varepsilon^{BC} \]
where
\[ R^D{_{ABC}}=\frac12 R\left( \delta^D{_B}g_{AC}- \delta^D{_C}g_{AB}\right) \]
hence
\[ R^D{_{ABC}}\varepsilon^{BC} = R \varepsilon^D{_A} \,. \]
Using the above formulae and the identity
\[ f_{||AB}\varepsilon^{AB} =0 \]
we can derive the following integrability condition:
\begin{equation}\label{laplasjanR}
\frac14 R^{||A}{_A} +2(R\omega^A)_{||A}=
 6f^2+ \frac38 R ( R-12\omega_A\omega^A ) \, .
\end{equation}
Equation (\ref{laplasjanR}) implies that there exists
non-empty open subset where $12\omega_A\omega^A > R > 0$.

\subsection{Transformation to linear problem}

Let us denote
\[ \Phi_A: = \frac{\omega_A}{\omega^B\omega_B} \,. \]
For any domain where $\omega^B\omega_B > 0$
equation (\ref{basic2}) implies
\begin{equation}\label{rotPhi} \Phi_{A||C}\varepsilon^{AC}
= \left(\frac{\omega_A}{\omega^B\omega_B}\right)_{||C}
\varepsilon^{AC} = 0 \end{equation}
which simply means that the one-form $\Phi_A\rd x^A$ is closed
and locally there exists coordinate $\Phi$ such that
\[ \rd\Phi= \Phi_A\rd x^A \, . \]
Moreover, from (\ref{basic2}) we get
\begin{equation}\label{divPhi} \Phi^A{_{||A}} = 1 \end{equation}
hence the potential $\Phi$ is a solution of the Poisson's equation:
\begin{equation}\label{lapPhi}
 \triangle\Phi = 1 \, .
\end{equation}
\underline{Remark} If we choose one isolated point where $\omega$
vanishes then for a given metric $g$ we have unique solution of the
above Laplace-Beltrami equation (Green function in the enlarged
sense). For more isolated points we can take linear combination of such
solutions. More precisely, let $G_{x_0}$ be a unique solution (for a given
metric $g$) of the equation (\ref{lapPhi}) on $S^2-\{x_0\}$.
If $c_0+ c_1+ \ldots + c_n = 1$ (where $c_i \in {\mathbb R}$)
then $\Phi= c_0 G_{x_0}+ c_1 G_{x_1}+ \ldots + c_n G_{x_n}$ is
a solution of (\ref{lapPhi}) on $S^2-\{ x_0, x_1, \ldots x_n \}$ and
$\omega$ vanishes at the points $x_0, x_1, \ldots x_n$.

\subsection{Solution of the problem with axial symmetry}
Let us consider axially symmetric two-metric on a sphere
in the following form:
\begin{equation}\label{gax} g=2m^2 \left[ A^{-1}(\theta)\rd\theta^2
+ A(\theta)\sin^2\theta\rd\phi^2 \right] \end{equation}
where $A: [0,\pi] \rightarrow {\mathbb R}$ is a positive smooth function with
boundary values $A(0)=A(\pi)=1$ and positive constant $m^2$ controls the
total volume of a sphere.
Eq. (\ref{gax}) implies that $\lambda=\sqrt{\det g_{AB}}=2m^2\sin\theta$.
From (\ref{divPhi}) we get
\[ \partial_A\left(\lambda\Phi^A\right)=\lambda \]
hence
\[ \lambda\Phi^\theta= -2m^2(\cos\theta + C) \]
 where $C$ is a constant of integration. Moreover, from
(\ref{rotPhi}) we obtain $\partial_\theta\Phi_\phi=0$ and
\[ \Phi_\phi=2m^2\alpha \]
with arbitrary constant $\alpha$. The equation (\ref{divom1}),
in terms of $\Phi^A$, takes the following form:
\begin{equation}\label{PhiK}
\partial_A\left(\frac{\lambda\Phi^A}{\Phi^B\Phi_B}\right)+
\frac{\lambda}{\Phi^B\Phi_B}-\lambda K =0 \, .
\end{equation}
The square of vector $\Phi^A$:
\[ \Phi^A\Phi_A= 2m^2\frac{(\cos\theta + C)^2+\alpha^2}{A\sin^2\theta}\, , \]
Gaussian curvature:
\[ \lambda K= - \frac12 \partial_\theta \left[ \frac1{\sin\theta}
  \partial_\theta \left( A\sin^2\theta\right) \right] \]
and the equation (\ref{PhiK}) imply that the function $A$ obeys
the following linear ODE:
\begin{equation}\label{ODE}
\frac{\rd}{\rd x}\frac{(x+C)y}{(x+C)^2+\alpha^2}+
\frac{y}{(x+C)^2+\alpha^2} +\frac12 \frac{\rd^2 y}{\rd x^2}=0
\end{equation}
where $x:=\cos\theta$ and $y:=A\sin^2\theta$. For $\alpha=0$ we get
\[ \frac{\rd^2}{\rd x^2} \left[ (x+C)y \right] =0 \]
with a general solution $y=\frac{ax+b}{x+C}$. However, in the case
$\alpha=0$
the function $A=\frac{y}{1-x^2}$ can not be regular at both points
$+1$ and $-1$ simultaneously. Nonexistence of regular
solutions for $\alpha=0$ confirms the main result
of \cite{CRT} because $\Phi_\phi=0$ gives $\omega_\phi=0$ which
 obviously implies $\rd \omega=0$.

For $\alpha\neq 0$ we take a new variable $t:=\frac{x+C}{\alpha}$
and the equation (\ref{ODE}) takes the form
\[ \frac{\rd}{\rd t}\left[ \frac{\rd}{\rd t}(ty)-\frac{2y}{1+t^2}
 \right] =0 \]
with the following general solution
\begin{equation}\label{sy} y=\frac{at+b(t^2-1)}{t^2+1} \end{equation}
with arbitrary constants $a,b$.
The solution (\ref{sy}) gives the following form of the function $A$:
\[ A = \frac{y}{1-x^2}=\frac{a\alpha(x+C)+b[(x+C)^2-\alpha^2]}{(1-x^2)
[(x+C)^2+\alpha^2]} \, .\]
Regularity of $A$ at $x=\pm 1$ implies that
$C^2=1-\alpha^2$ (hence $0< |\alpha | \leq 1$) and
$\frac{a}{b}\alpha+2C=0$ which gives
\[ A =\frac{-b}{[(x+C)^2+\alpha^2]} \, .\]
Moreover, $A(\pm 1)=1$ implies $b=-2$, $\alpha=1$, $C=0$ hence
\[ A=\frac2{1+x^2}=\frac2{1+\cos^2\theta} \, ,\]
and finally
\begin{equation}\label{gKerr}
g=2m^2 \left[ \frac{1+\cos^2\theta}2\rd\theta^2
+ \frac{2\sin^2\theta}{1+\cos^2\theta}\rd\phi^2 \right] \end{equation}
and
\begin{equation}\label{omKerr}
 \omega^\theta=-\frac{\sin\theta\cos\theta}{m^2(1+\cos^2\theta)^2} \, ,\quad
   \omega^\varphi=\frac{1}{2m^2(1+\cos^2\theta)} \, ,
\end{equation}
which corresponds to extremal Kerr with mass $m$ and angular momentum $m^2$.

It is worth to notice that the solution (\ref{omKerr}) in terms of $\Phi_A$
has a simple and natural form. More precisely, equations (\ref{rotPhi})
and (\ref{divPhi}) extended through the ``poles'' are the following:

\begin{equation}\label{rotPhiKerr} \Phi_{A||C}\varepsilon^{AC}
= 4\pi m^2 \left( {\boldsymbol\delta}_{\theta=\pi} -
{\boldsymbol\delta}_{\theta=0} \right) \, ,
\end{equation}

\begin{equation}\label{divPhiKerr} \Phi^A{_{||A}} = 1 -
4\pi m^2 \left( {\boldsymbol\delta}_{\theta=\pi} +
{\boldsymbol\delta}_{\theta=0} \right) \, ,
\end{equation}
where by $\boldsymbol\delta_p$ we denote a Dirac delta at point $p$
and $8\pi m^2$($=\int\lambda$) is a total volume of the sphere (\ref{gKerr}).

Let $G_p$ be a Green function satisfying
\begin{equation}\label{lapGp}
\left\{ \begin{array}{l}
 \triangle G_p = 1 -8\pi m^2 \boldsymbol\delta_p \, , \\[1ex]
 \int\lambda G_p = 0 \, .
 \end{array} \right.
\end{equation}
The potentials $\Phi$, $\tilde\Phi$ for the covector field $\Phi_A$
defined (up to a constant) as follows
\begin{equation}\label{potPhi}
    \Phi_A =  \partial_A\Phi+\varepsilon_A{^B}\partial_B\tilde\Phi
\end{equation}
take a simple form
\begin{eqnarray}
  \Phi &=& \frac12 (G_{\theta=0}+G_{\theta=\pi} ) \\
  \tilde\Phi &=& \frac12 (G_{\theta=0}-G_{\theta=\pi} )
\end{eqnarray}
because equations (\ref{rotPhiKerr}), (\ref{divPhiKerr}) and
(\ref{potPhi}) imply
\begin{eqnarray}
  \triangle \Phi &=& 1 -4\pi m^2 \left( {\boldsymbol\delta}_{\theta=\pi} +
{\boldsymbol\delta}_{\theta=0} \right) \, , \\
  \triangle \tilde\Phi &=& 4\pi m^2 \left( {\boldsymbol\delta}_{\theta=\pi} -
{\boldsymbol\delta}_{\theta=0} \right) \, .
\end{eqnarray}
Moreover, the Green functions for extremal Kerr (\ref{gKerr}) are given
in the explicit form:
\begin{eqnarray}
  G_{\theta=0} &=& 4m^2\left[ \frac12\sin^2\frac\theta{2}+\frac18\sin^2\theta
  -\log(\sin\frac\theta{2}) +\frac13 \right] \, , \\
  G_{\theta=\pi} &=& 4m^2\left[ \frac12\cos^2\frac\theta{2}+\frac18\sin^2\theta
  -\log(\cos\frac\theta{2}) +\frac13 \right] \, .
\end{eqnarray}
\section{Conclusions}
We have discussed some geometric consequences of the {\it basic equation}
(\ref{basic}) appearing in the context of
Kundt's class metrics and degenerate (extremal) horizons.
We have obtained several important results like
topological rigidity of the horizon (Section 2 and 3),
integrability conditions 
and transformation to
linear problem which simplifies the proof of
the uniqueness of extremal Kerr for axially symmetric horizon
(Section 4). However, the problem of the existence of non-symmetric solutions
to the {\it basic equation} remains opened.
\appendix
\section{Extremal Kerr}
For extremal Kerr 
we have
\[
h=\left(1-\frac12\sin^2\theta\right)^2\rd\theta^2 +\sin^2\theta\rd\varphi^2
\, ,\quad
\lambda_h:=\sqrt{\det h_{AB}}=\frac{1+\cos^2\theta}{2}\sin\theta\]
\[  \exp(-2\alpha)=\frac{4m^2}{1+\cos^2\theta}\, ,\quad
K_h= \frac{4\sin^2\theta}{(1+\cos^2\theta)^3} \]
\[ \omega^\theta=-\frac{\sin\theta\cos\theta}{m^2(1+\cos^2\theta)^2} \, ,\quad
   \omega^\varphi=\frac{1}{2m^2(1+\cos^2\theta)} \, , \quad
 \|\omega\|=\frac{\sin\theta}{2m^2\sqrt{1+\cos^2\theta}} \]
Equation (\ref{hom1}) for extremal Kerr takes a simple form:
\[
\lambda_h\omega^\theta\|\omega\|^{-2}= -2m^2\cos\theta \, , \quad
 \lambda_h\exp(-2\alpha)=\lambda_h\frac{\sqrt{K_h}}{\|\omega\|}
 =2m^2\sin\theta \, .\]


\begin{thebibliography}{99}
\bibitem{ABL} A.~Ashtekar, C.~Beetle and J.~Lewandowski,
  Class. Quantum Grav. {\bf 19},   1195--1225  (2002)
\bibitem{CRT} P. Chru\'sciel, H.S. Real and P. Tod, {\it On non-existence
 of static vacuum black holes with degenerate components of the event horizon},
  NI05066-GMR
\bibitem{JKC}
J.~Jezierski, J.~Kijowski, and E.~Czuchry, Rep. Math. Phys. {\bf
46}, 397 (2000).

\bibitem{jjPRD02} J.~Jezierski, J.~Kijowski, and E.~Czuchry,
 Phys. Rev. D {\bf 65}, 064036 (2002)
\bibitem{exact}
 H. Stephani et al., \emph{Exact solutions of Einstein's
 field equations}, 2nd ed., University Press (Cambridge 2003)
\bibitem{LP} J. Lewandowski and T. Paw{\l}owski, Classical and Quantum Gravity
{\bf 20} (2003), 587--606
 \bibitem{LPJJ} T. Paw{\l}owski, J. Lewandowski and J. Jezierski,
 {\it Spacetimes foliated by Killing horizons},
  gr-qc/0306107, Classical and Quantum Gravity {\bf 21} (2004) 1237--1251

\end{thebibliography}
\end{document}